\documentclass[preprint,aps,nofootinbib,a4paper,superscriptaddress,longbibliography,amsfonts,amssymb,amsmath,titlepage]{revtex4-1}
\usepackage{IEEEtrantools}
\usepackage{microtype}
\usepackage[toc,title,page,titletoc]{appendix}
\usepackage{amsmath}
\usepackage{amsthm}
\usepackage{mathrsfs}
\usepackage{enumitem}
\usepackage[T1]{fontenc}
\usepackage{xfrac}

\font\ec=ecrm0800 at 12pt
\def\thorn{\hbox{\ec\th}}
\def\eth{\hbox{\ec\dh}}

\newtheorem{theorem}{Theorem}
\newtheorem*{corollary}{Corollary}
\newtheorem{lemma}[theorem]{Lemma}

\newcommand\widebar[1]{\mathop{\overline{#1}}}
\providecommand{\abs}[1]{\left\lvert#1\right\rvert}
\providecommand{\norm}[1]{\left\lVert#1\right\rVert}

\providecommand{\swss}[2]{S_{m, \lambda}^{[#1], (#2)}}
\providecommand{\swsl}[2]{\lambda_m^{[#1], (#2)}}
\providecommand{\swsc}[2]{C_{m, \lambda}^{[#1], (#2)}}

\def\half{\frac{1}{2}}
\renewcommand{\Re}{\operatorname{Re}}
\renewcommand{\Im}{\operatorname{Im}}

\usepackage{lmodern}
\usepackage[a4paper]{geometry}
\usepackage{hyperref}
\usepackage{mathbbol}

\begin{document}

\title{On the radiation gauge for spin-1 perturbations in Kerr-Newman spacetime}

\author{Stefan Hollands}
\email{stefan.hollands@uni-leipzig.de}
\affiliation{Institut f\"ur Theoretische Physik, Universit\"at Leipzig and \\
Max-Planck-Institute MiS, Leipzig}

\author{Vahid Toomani}
\email{vahid.toomani@mis.mpg.de}
\affiliation{Institut f\"ur Theoretische Physik, Universit\"at Leipzig and \\
Max-Planck-Institute MiS, Leipzig}

\begin{abstract}
We extend previous work \cite{Green:2019nam} to the case of Maxwell's equations with a source.
Our work shows how to construct a vector potential for the Maxwell field on the Kerr-Newman background in a radiation gauge. The vector potential has a ``reconstructed'' term obtained from a Hertz potential solving Teukolsky's equation with a source, and a ``correction'' term which is obtainable by a simple integration along outgoing principal null rays. The singularity structure of our vector potential is discussed in the case of a point particle source.
\end{abstract}

\maketitle
\tableofcontents

\section{Introduction}

Perturbed black hole spacetimes play an important role in gravitational physics, for instance when modeling the ringdown phase of a black hole merger or to describe accurately an extremal mass ratio inspiral. 
In so far as linear perturbations are concerned, the Teukolsky formalism \cite{Teukolsky:1973ha, Teukolsky:1972my} has many practical and conceptual advantages. As originally formulated, it yields the perturbations of the Weyl tensor $C_{abcd}$ (spin 2) or the Maxwell tensor $F_{ab}$ (spin 1). 
However, there are cases in which one is interested in the corresponding potentials, i.e. the metric perturbation, $h_{ab}$ (spin 2) or the perturbation of the vector potential, $A_a$ (spin 1). For spin 2, a prominent example is the self-force approach \cite{Quinn:1996am, Mino:1996nk} to extreme mass ratio inspiral. Actually, in this context, also the spin 1 equation in effect has to be solved when converting to a convenient gauge, see below.\par
Soon after Teukolsky's work \cite{Teukolsky:1973ha, Teukolsky:1972my}, it was pointed out by \cite{Chrzanowski:1975wv, Kegeles:1979an} that an equation very similar to Teukolsky's in fact also applies directly to the perturbations of the metric resp. the electromagnetic vector potential -- in a sense.
In their approach, the metric perturbation resp. vector potential are obtained by applying a certain second- resp. first order partial differential ``reconstruction'' operator, ${\mathcal S}^\dagger$, to a scalar ``pre-potential'' satisfying Teukolsky's master equation. But a major shortcoming of their method is that, while it applies straightforwardly in the source free case, it is easily seen to be inconsistent for generic sources, such as e.g. that generated by a massive resp. charged point particle. For the same reason, their method is unsuitable as it stands for calculating higher order metric perturbations, which is a practically relevant problem for the extreme mass ratio inspiral problem \cite{Pound:2017psq, Harte:2011ku, Detweiler:2011tt, Pound:2013iba, Gralla:2012db, Pound:2012dk}.\par
In a recent paper \cite{Green:2019nam} co-authored by one of us, a resolution of this problem was found in the case of spin 2. It consists in adding a very special ``corrector'' to the ansatz of CCK \cite{Chrzanowski:1975wv, Kegeles:1979an}, and is henceforth referred to as ``corrected CCK'' (3CK). This corrector is determined by simple integrations along integral curves of one of the algebraically special null fields, $l^a$, of Kerr. 
The main purpose of this paper is to present a version of this construction for spin 1. The procedure that we develop is qualitatively similar to spin 2 but somewhat simpler.\par
While the electromagnetic self-force $F^{ab}J_b$ does not require the knowledge of the vector potential but only of $F_{ab}$, the determination of the vector potential is nevertheless  of interest as a simpler toy model for the case of gravitational perturbations, and would in particular be relevant for calculating higher order self-force corrections in theories wherein charged fields couple to the vector potential\footnote{For recent progress on second order self-force, see e.g. \cite{Pound:2017psq}.}. \par
More importantly, the equation for the spin 1 potential has to be solved in effect anyway -- even in the case of gravitational perturbations alone -- when converting the spin 2 metric perturbation $h_{ab}$ obtained in a singular gauge (e.g. the radiation gauge obtained via the 3CK method of \cite{Green:2019nam}) to the Lorenz gauge. In this problem, we seek a $\xi_a$ such that $h_{ab} + \nabla_a \xi_b + \nabla_b \xi_a$ satisfies the Lorenz gauge condition. It is not hard to see that finding the explicitly required gauge vector field $\xi_a$ can be reduced to the problem of solving the sourced Maxwell equation \cite{WIP}, for which we provide a method in this work. The Lorenz gauge is of special interest because it is known to be more regular than the 3CK form of the metric \cite{Pound:2013faa}, in the sense that it contains singularities only on the worldline itself, rather than also away from the worldline. For this reason, it is expected to be important for obtaining the second order metric perturbations, as required in higher order self-force calculations, in a framework taking full advantage of the Teukolsky formalism.\par
The plan of this paper is as follows. In sec. \ref{sec:2} we present the basic equations in GHP-form \cite{Geroch:1973am}. 
In sec. \ref{sec:3}, we first analyze to what extent a general solution $A_a$ to the homogeneous Maxwell equations can be written up to gauge in CCK form on the Kerr-Newman spacetime. Our main result is that the space of CCK perturbations modulo gauge is dense in a natural sense in the space of all vector potentials having a suitable fall-off near spatial infinity. While this result is analogous to a corresponding result in our previous paper \cite{Green:2019nam}, we point out an important and non-trivial technical advance in our argument which allows us to discard any condition on the behavior of the vector potential at the horizon. Our result hinges on the Hamiltonian formulation of Maxwell's equation, the Starobinsky-Teukolsky identities, and rather non-trivial decay results obtained recently by \cite{Ma:2017yui}, see also \cite{Dafermos:2016uzj, Dafermos:2017yrz, Andersson:2019dwi} for related results.
Then we show in sec. \ref{sec:4} that the forward solution to the sourced Maxwell equations for $A_a$ can be obtained in terms of a 3CK ansatz.
The correction piece is obtained via an explicit and simple integration scheme. Finally, we briefly discuss this scheme in sec. \ref{sec:5} for the case of a point source on a geodesic. The purpose is to point out the relation between our 3CK perturbations to the ``half-string'' solutions investigated by \cite{Pound:2013faa}, and to argue rigorously that, after passing to a Lorenz gauge, the half-string singularities are absent (the same argument would apply to spin 2). 

\section{Maxwell's equations in GHP formalism}
\label{sec:2}

In this paper we study the Maxwell equation
\begin{equation}
(\mathcal{E} A)_a \equiv \nabla^b F_{ab} = 4\pi J_a
\end{equation}
on the Kerr-Newman spacetime\footnote{Most of our formalism in fact also applies to general type D metrics.}.
In order to take advantage of the algebraically special properties of the Kerr-Newman metric, a Petrov type D spacetime, it is customary to employ the Newman-Penrose, or its cousin, the Geroch-Held-Penrose (GHP) formalism \cite{Geroch:1973am}. In either formalism, the metric is encoded in a suitable null tedrad $l^a, n^a, m^a, \widebar{m}^a$. We take $l^a, n^a$ adapted to the principal null directions of the spacetime; see appendix \ref{KerrNewmanQuants} for a particular choice. The GHP formalism has a geometrical basis rendering the basic operators invariant under the remaining tetrad transformations (``gauge''). In fact, the physical fields become in this formalism sections in certain line bundles characterized by weights. For such a section we write $\Phi \circeq \left\{ p, q \right\}$ if it is a section in the line bundle characterized by the boost weight $r = \sfrac{(p+q)}{2}$ and the spin weight $s = \sfrac{(p-q)}{2}$, see e.g. \cite{Aksteiner:2010rh, ehlers} for details. The GHP calculus has all simplifications arising from the type D property built into the basic identities. These simplifications (see appendix \ref{KerrNewmanQuants}) will be exploited in many places in the sequel without further notice.

\par
Recall that Maxwell's equations are in GHP form:
\begin{subequations}
\begin{align}
\label{MaxwelleqJl}
\left( \eth' - \tau' \right) \phi_0 - \left( \thorn - 2 \rho \right) \phi_1 =& 2\pi J_l,\\
\label{MaxwelleqJn}
\left( \thorn' - 2 \rho' \right) \phi_1 - \left( \eth - \tau \right) \phi_2 =& 2\pi J_n,\\
\label{MaxwelleqJm}
\left( \thorn' - \rho' \right) \phi_0 - \left( \eth - 2 \tau \right) \phi_1 =& 2\pi J_m,\\
\label{MaxwelleqJmbar}
\left( \eth' - 2 \tau' \right) \phi_1 - \left( \thorn - \rho \right) \phi_2 =& 2\pi J_{\widebar{m}}.
\end{align}
\end{subequations}
Here, $\phi_2 \circeq \left\{ -2, 0 \right\}$, $\phi_1 \circeq \left\{ 0, 0 \right\}$ and  $\phi_0 \circeq \left\{ 2, 0 \right\}$ are the Maxwell scalars, defined as:
\begin{subequations}
\begin{equation}
\phi_2 = F_{\widebar{m} n} = \left( \eth' - \bar{\tau} \right) A_n - \left( \thorn' - \bar{\rho}' \right) A_{\widebar{m}},
\end{equation}
\begin{multline}
2 \phi_1 = F_{l n} + F_{\widebar{m} m} = - \left( \thorn' + \rho' - \bar{\rho}' \right) A_l + \left( \thorn + \rho - \bar{\rho} \right) A_n \\
+ \left( \eth' + \tau' - \bar{\tau} \right) A_m - \left( \eth + \tau - \bar{\tau}' \right) A_{\widebar{m}},
\end{multline}
\begin{equation}
\phi_0 = F_{l m} = \left( \thorn - \bar{\rho} \right) A_{m} - \left( \eth - \bar{\tau}' \right) A_l.
\end{equation}
\end{subequations}
The symbols $\thorn$, $\thorn'$, $\eth$, $\eth'$ represent suitably covariantized directional derivatives in the $l$, $n$, $m$, $\widebar{m}$-directions, and $\rho \circeq \left\{ 1, 1 \right\}$, $\rho' \circeq \left\{ -1, -1 \right\}$, $\tau \circeq \left\{ 1, -1 \right\}$, $\tau' \circeq \left\{ -1, 1 \right\} $ are the non-vanishing optical scalars in the spacetime. 
The Maxwell scalar $\phi_0$ can be viewed as defining an operator $\mathcal{T}$, given in GHP form by:
\begin{subequations}
\begin{align}
\label{Odef}
\mathcal{O} \Phi =& \left[ \left( \thorn - \bar{\rho} - 2\rho \right) \left( \thorn' - \rho' \right) - \left( \eth - \bar{\tau}' - 2\tau \right) \left( \eth' - \tau' \right) \right] \Phi,\\
\label{Tdef}
\mathcal{T} A =& \left( \eth - \bar{\tau}' \right) A_l - \left( \thorn - \bar{\rho} \right) A_m = - \phi_0,\\
\label{Sdef}
\mathcal{S} J =& \half \left[ \left( \eth - \bar{\tau}' - 2\tau \right) J_l - \left( \thorn - \bar{\rho} - 2\rho \right) J_m  \right].
\end{align}
\end{subequations}
The other operators $\mathcal O, \mathcal S$ have been defined so as to satisfy the fundamental operator identity in the Teukolsky formalism \cite{Wald:1978vm}.
\begin{equation}
\mathcal{SE} = \mathcal{OT},
\end{equation}
which implies that if $A_a$ is a solution to the Maxwell's equation, $\phi_0$ is a solution to $\mathcal{O} \phi_0 = - 4\pi \mathcal{S}J$, which is the Teukolsky form of the Maxwell's equations. Applying the adjoint of the fundamental operator equation to $\Phi \circeq \left\{ -2, 0 \right\}$, one obtains
\begin{equation}
\mathcal{E}^\dagger \mathcal{S}^\dagger \Phi = \mathcal{T}^\dagger \mathcal{O}^\dagger \Phi. 
\end{equation}
See appendix \ref{teukolskyoperators} for the concrete expressions of the adjoint operators. In particular, as observed by \cite{Chrzanowski:1975wv,Kegeles:1979an}, if $\mathcal{O}^\dagger \Phi=0$, then $A_a \equiv \Re(\mathcal{S}^\dagger \Phi)_a$ is a real-valued solution to the Maxwell's equations with vanishing source. Such a $\Phi \circeq \left\{ -2, 0 \right\}$ is called a ``Hertz potential''. A vector potential of the form $A_a = \Re(\mathcal{S}^\dagger \Phi)_a$ is said to be in ``CCK form'' after the inventors of this ansatz \cite{Chrzanowski:1975wv,Kegeles:1979an}. Note that any CCK vector potential automatically has $A_l \equiv A_a l^a = 0$. Such a gauge is referred to traditionally as {\em ``(ingoing) radiation gauge.''}

\section{Homogeneous Maxwell's equations}
\label{sec:3}

In this section, we analyze the question to what extent ``every'' solution $A_a$ to the source free Maxwell equations can be written in CCK form $A_a = 2 \Re(\mathcal{S}^\dagger \Phi)_a$, up to gauge on the Kerr-Newman spacetime, i.e. whether we can find $\Phi$ and $\chi$, with $\mathcal{O}^\dagger \Phi=0$, such that
\begin{equation}
\begin{split}
A_a =& \Re(\mathcal{S}^\dagger \Phi)_a + \nabla_a \chi \\
=& \half \Re \left\{ \left[ m_a \left( \thorn + \rho \right) - l_a \left( \eth + \tau \right) \right] \Phi \right\} + \nabla_a \chi,
\end{split}
\end{equation}
using the expression for $\mathcal{S}^\dagger$ as given in Appendix \ref{teukolskyoperators} in the second line.
Of course, we should specify what ``every'' means, i.e. regularity, fall-off at infinity etc. Roughly speaking, we will show that the set of vector potentials $A_a$ which can be written in this form are dense in a suitable norm in a vector space of distributional solutions characterized by a vanishing induced electric charge, $\dot Q$, and a falloff $o(r^{-3/2})$ at spatial infinity of the electric
field and magnetic potential. \par
As in a previous work by one of us \cite{Green:2019nam}, we will employ a Hamiltonian formalism. We consider a Cauchy surface $\Sigma$ stretching between the bifurcation surface of the black hole and spatial infinity. Any sufficiently smooth solution is specified up to gauge by its initial data
\begin{equation}
\label{pairdef}
a = (\boldsymbol{A}_a, \boldsymbol{E}^a) = ({h_a}^b A_b |_{\Sigma}, F^{b a} \nu_b |_{\Sigma}),
\end{equation}
where $\nu^a$ is the unit forward time-like normal to $\Sigma$ and where $h_{ab} = g_{ab}-\nu_a\nu_b$ the induced metric. The first entry is the magnetic potential and the second the electric field on $\Sigma$. As usual, we have the Gauss constraint $D_a \boldsymbol{E}^a = 0$, where $D_a$ is the covariant derivative of $h_{ab}$. We also need to equip the space of solutions with a topology. The simplest conceivable one is $L^2\times L^2$ space with the the inner product,
\begin{equation}
\label{innerproduct}
\left\langle a, a' \right\rangle_\Sigma = - \int_\Sigma \sqrt{h} (\boldsymbol{A}_a \boldsymbol{A}'^{a} + \boldsymbol{E}_a \boldsymbol{E}'^{a}) .
\end{equation}
In addition  we equip our space with the usual symplectic form $\Omega$ :
\begin{equation}
\label{maxwellsymplecticform}
\Omega[a, a']  = \int_\Sigma \sqrt{h} ( \boldsymbol{A}^a \boldsymbol{E}'_a - \boldsymbol{E}^a \boldsymbol{A}'_a ).
\end{equation}
$\Omega$ induces a linear anti-hermitian unitary operator $\tilde{\Omega}$ on the $L^2\times L^2$. 
From that fact, it is very easy to prove the following elementary facts:

\begin{lemma}
Let $\mathscr{A}$ and $\mathscr{B}$ be linear subspaces of $L^2 \times L^2$, $\mathscr{A}^\perp$ and $\mathscr{B}^\perp$ be symplectic orthogonal complements of $\mathscr{A}$ and $\mathscr{B}$, and $\widebar{\mathscr{A}}$ the closure of $\mathscr{A}$ w.r.t. inner product. Then:

\begin{enumerate}
\item[i.]  $\widebar{\mathscr{A}}^\perp = \mathscr{A}^\perp$,
\item[ii.\label{secondpartoffirsttheorem}] $\mathscr{A}^\perp$ is a closed subspace,
\item[iii.] ${\mathscr{A}^\perp}^\perp = \widebar{\mathscr{A}}$, and
\item[iv.] $\left( \mathscr{A} + \mathscr{B} \right)^\perp = \mathscr{A}^\perp \cap \mathscr{B}^\perp$.
\end{enumerate}
\end{lemma}

Following \cite{Green:2019nam}, which builds on \cite{Hollands:2012sf, Prabhu:2018jvy}, we now introduce certain subspaces of $L^2 \times L^2$. Let $\mathscr{W}$ be the subspace of all gauge perturbations, $A_a = \nabla_a \chi$ generated by a smooth function $\chi$ that is constant in a neighborhood of spacial infinity $i^0$ and on the bifurcation surface $\mathscr{S}$. In addition, let $\mathscr{V}$ be the symplectic orthogonal complement of $\mathscr{W}$, i.e., $\mathscr{V} = \mathscr{W}^\perp$. We can introduce within $\mathscr{V}$ the subspace $\mathscr{V} \cap \mathscr{U}$, where $\mathscr{U} \subset C^\infty \times C^\infty$ is a certain intersection of weighted Sobolev spaces. As in section 4.1 of \cite{Green:2019nam} one may show that: (i) Initial data from $\mathscr{V}$ satisfy the constraint $D_a \boldsymbol{E}^a=0$ in the distributional sense, (ii) have $\dot{Q} = 0$, where
\begin{equation}
\dot{Q} = \frac{1}{2\pi} \int_{i^0} \star F
\end{equation}
is the charge carried by $A_a$, (iii) $\mathscr{U} \cap \mathscr{V}$ is dense in the closed subspace $\mathscr{V}$, and (iv) satisfy $\boldsymbol{A}_a, \boldsymbol{E}^a = o(r^{-3/2})$ in an asymptotically Cartesian coordinates system as $r \to \infty$, with derivatives falling off faster by corresponding powers of $r^{-1}$.\par
Using our gauge freedom, we could impose that the electro-static potential vanishes at the horizon, or more precisely that $A_n = A_l = 0$ on $\mathscr{S}$, but $A_m, A_{\widebar{m}}$ cannot be set to any particular values on $\mathscr{S}$. Furthermore, by an analog of the density result (\ref{homogeneoustheorem}), which uses gluing methods, we may replace $\mathscr{U}$ by the space of smooth pairs $(\boldsymbol{A}_a, \boldsymbol{E}^a)$ that vanish outside a large ball.\par
Finally, let $\mathscr{Y}$ be the subspace of potentials of the form $\Re \left\{ \mathcal{S}^\dagger \Phi \right\}$ where $\Phi$ is a smooth GHP scalar of weight $\left\{-2, 0 \right\}$ such that, $\mathcal{O}^\dagger \Phi = 0$, and such that $\Phi=0=\dot \Phi$ outside a large ball on $\Sigma$. By definition, every $A_a$ in $\mathscr{Y}$ has vanishing electric charge $\dot Q=0$. The definitions imply that
\begin{equation}
\mathscr{Y} \subset \mathscr{V}.
\end{equation}
If $A_a$ is a distributional solution to the homogeneous Maxwell's equation, then we say $A_a \in \mathscr{V}$ if this is true for their initial data, etc. Momentarily, we will also show that 
\begin{equation}
\left( \mathscr{U} \cap \mathscr{V} \right) \cap \mathscr{Y}^\perp \subseteq \mathscr{W},
\end{equation}
at least for sufficiently slowly rotating Kerr black holes.
This gives 
\begin{multline}
\mathscr{U} \cap \left( \mathscr{W} + \mathscr{Y} \right)^\perp = \left( \mathscr{U} \cap \mathscr{V} \right) \cap \mathscr{Y}^\perp \subseteq \mathscr{W} \Longrightarrow \mathscr{V} = \mathscr{W}^\perp \subseteq \left( \left( \mathscr{W} + \mathscr{Y} \right)^\perp \cap \mathscr{U} \right)^\perp \\
= \left( \overline{\left( \mathscr{W} + \mathscr{Y} \right)^\perp \cap \mathscr{U}} \right)^\perp = \left( \left( \mathscr{W} + \mathscr{Y} \right)^\perp \right)^\perp = \overline{\mathscr{W} + \mathscr{Y}}.
\end{multline}
In other words, $\mathscr{V} \subseteq \overline{\mathscr{W} + \mathscr{Y}}$. Therefore, we have the following density result:

\begin{theorem}
\label{homogeneoustheorem}
Any distributional solution $A_a$ to the homogeneous Maxwell's equations on the exterior of the Kerr black hole with $a \ll M$, with initial in the space $\mathscr{V}$ can be approximated to arbitrary precision $\epsilon>0$ by a solution of the form $\Re(\mathcal{S}^\dagger \Phi)_a + \nabla_a \chi$, where $\chi$ is smooth, constant near infinity, and vanishing at the bifurcation surface, and where $\Phi$ is a Hertz-potential, $\mathcal{O}^\dagger \Phi = 0$, with smooth initial data that vanish near the horizon and near infinity, in the sense that
\begin{equation}
\norm{A - {\rm d} \chi -  \Re ( \mathcal{S}^\dagger \Phi ) }_\Sigma < \epsilon,
\end{equation}
in which $\norm{.}_\Sigma$ denotes the $L^2(\Sigma \times \Sigma)$ norm on initial data \eqref{innerproduct}.
\end{theorem}

\noindent
\textbf{Remark:} The condition $a \ll M$ is necessary in order to apply the decay results of \cite{Ma:2017yui}. If these results could be demonstrated for all $a<M$, or for Kerr-Newman, then the above theorem would hold accordingly.

\medskip
\noindent
Thus we need:
\begin{lemma}
Under the conditions of the theorem, $\left( \mathscr{U} \cap \mathscr{V} \right) \cap \mathscr{Y}^\perp \subset \mathscr{W}$.
\end{lemma}
\begin{proof}
Let $A_a$ be some potential in $\left( \mathscr{U} \cap \mathscr{V} \right) \cap \mathscr{Y}^\perp$. Then by construction $A_a$ is a smooth solution to the homogeneous Maxwell equations, has vanishing initial data near infinity, and for any smooth Hertz potential $\Phi$ with vanishing  initial data near infinity, 
\begin{equation}
\label{eq:1}
\Omega[A, \Re(\mathcal{S}^\dagger \Phi)] = 0. 
\end{equation}
We must show that $A_a = \nabla_a \chi$ with $\chi$ vanishing near $i^0$ and constant on $\mathscr{S}$, so in particular 
$A_m=A_{\widebar{m}}=0$ on $\mathscr{S}$. Our strategy for showing this is as follows. 
First, we show $\phi_0=0$. Then, we argue that also $\phi_2=0$, and then that $\phi_1=0$, which shows that $A_a$ is pure gauge. 
Finally, we demonstrate the properties of the gauge function, $\chi$. \par
First, note that \eqref{eq:1} holds also for $i\Phi$, and since $A_a$ is real, we may drop the ``Re''. 
By the formulas \eqref{pivalue}, \eqref{wvalue}, \eqref{wsigmapitidentity}, \eqref{starH}, we find
\begin{equation}
\label{eq:2}
\int_\Sigma \star w[A, \mathcal{S}^\dagger \Phi] = \int_\Sigma \star \pi[\mathcal{T} A, \Phi] + \int_{\mathscr{S}} \star H[A, \Phi].
\end{equation}
Furthermore,
\begin{equation}
\label{eq:3}
\Omega[A, A'] = \int_\Sigma \star w[A,A'],  
\end{equation}
for any pair of solutions $A_a, A'_a$, in particular $A_a' = (\mathcal{S}^\dagger \Phi)_a$. This must hold for an arbitrary choice of $\Phi, \thorn \Phi$ on $\Sigma$. Therefore, the definition \eqref{pivalue} implies that the initial data for  $\mathcal{T} A = - \phi_0 = 0$ on $\Sigma$ are trivial. Further, given that $\mathcal{O} \phi_0 = 0$, we conclude that $\phi_0 = 0$ on the whole exterior region $\mathscr{M}$ in view of the usual existence and uniqueness properties of solutions to hyperbolic partial differential equations such as these.\par
We next implement the Starobinsky-Teukolsky identities \eqref{starobinskyidentities} with the intention to show that also $\phi_2 = 0$.
\begin{subequations}
\label{starobinskyidentities}
\begin{align}
\label{starobinskyidentities1}
\left( \thorn - 2\rho \right)^2 \phi_2 =& \left( \eth' - 2\tau' \right)^2 \phi_0,\\
\label{starobinskyidentities2}
\left( \eth - 2\tau \right)^2 \phi_2 =& \left( \thorn' - 2\rho' \right)^2 \phi_0.
\end{align}
\end{subequations}
By setting $\phi_0$ to zero in \eqref{starobinskyidentities1} we learn using the identity $\thorn \rho = \rho^2$ that
\begin{equation}
\label{eq:ph2}
\frac{1}{\rho^2} \phi_2 = \alpha^\circ_0 + \frac{1}{\rho} \alpha^\circ_1. 
\end{equation}
Here, a degree mark as in $\alpha^\circ_i$ means that a GHP scalar is annihilated by $\thorn$. To deal with such objects efficiently, it is convenient to use a calculus introduced by Held \cite{held1974formalism, held1975formalism}. Held's operators are on weight $\left\{ p, q \right\}$ GHP quantities,
\begin{equation}
\label{eq:Hops}
\begin{split}
\tilde{\thorn}' &=\thorn' - \bar{\tau} \tilde{\eth} - \tau \tilde{\eth}' + \tau \bar{\tau} \left( \frac{q}{\bar \rho} + \frac{p}{\rho} \right) + \half \left( \frac{q \bar \Psi_2}{\bar \rho} + \frac{p \Psi_2}{\rho} \right), \\
\tilde{\eth} &=\frac{1}{\bar \rho} \eth + \frac{q\tau}{\rho}, \\
\tilde{\eth}' &= \frac{1}{\rho} \eth' + \frac{p\bar{\tau}}{\bar \rho}.
\end{split}
\end{equation} 
The key feature that makes these operators useful is that they commute through $\thorn$ when acting on any GHP scalar $\alpha^\circ$, so e.g. $\left[ \thorn, \tilde{\eth} \right] \alpha^\circ = 0$. In the Kinnersley frame and advanced Kerr-Newman coordinates, see appendix \ref{KerrNewmanQuants}, the Held operators $\tilde{\eth}, \tilde{\eth}'$ and $\tilde{\thorn}'$ are
\begin{equation}
\label{eq:Hop}
\begin{split}
\tilde{\eth} & = -\frac{1}{\sqrt{2}} \left( \frac{\partial}{\partial \theta} + i \csc \theta \frac{\partial}{\partial \varphi} + ia \sin \theta \frac{\partial}{\partial u} - \half (p-q) \cot \theta \right) , \\
\tilde{\eth}' & = -\frac{1}{\sqrt{2}} \left( \frac{\partial}{\partial \theta} - i \csc \theta \frac{\partial}{\partial \varphi} - ia \sin \theta \frac{\partial}{\partial u} + \half (p-q) \cot \theta \right), \\
\tilde{\thorn}' & = \frac{\partial}{\partial u} ,
\end{split}
\end{equation}
and when acting on a ``mode'' in this frame they are basically equal to Chandrasekhar's operators $\mathcal{L}_s, \mathcal{L}_s^\dagger$, with $s=\sfrac{(p-q)}{2}$, see \cite{Chandrasekhar:1984siy}. On the other hand, by the Starobinsky-Teukolsky identities \eqref{starobinskyidentities2} and $\left[ \thorn, \tilde{\eth} \right] \alpha^\circ_1 = 0$ we have,
\begin{equation}
\label{relationsforalphas}
\tilde{\eth}^2 (\frac{1}{\rho^2} \phi_2) = 0 \Longrightarrow \tilde{\eth}^2 \alpha^\circ_0 = \tilde{\eth}^2 \alpha^\circ_1 = 0.
\end{equation}
The results by \cite{Ma:2017yui} next show that $\phi_2$ is an $L^2$-function of $u$ in any strip $r \in (r_0, r_1)$ in the exterior. This and the special form of $\phi_2$, \eqref{eq:ph2} implies that $\alpha^\circ_i(u, \theta, \varphi^*)$ are also $L^2$-functions of $u$. Then we can safely take Fourier-transforms of $\alpha^\circ_i$ in $L^2$ and expand in terms of spin-weighted harmonics $\swss{\pm1}{a \omega}(\theta)$ (see \cite{Chandrasekhar:1984siy} for details), to obtain
\begin{equation}
\alpha^\circ_i = \sum_{l,m} \int_{-\infty}^\infty {\rm d} \omega \; \hat{\alpha}_{i l m}(\omega) \swss{-1}{a \omega}(\theta) e^{-i\omega u + im \varphi}.
\end{equation}
Then equations \eqref{relationsforalphas} are equivalent to
\begin{equation}
\begin{split}
2 \tilde{\eth}^2 \alpha^\circ_i =& \sum_{l, m} \int_{-\infty}^\infty{\rm d} \omega \; \hat{\alpha}_{i l m}(\omega) \, \mathcal{L}^\dagger_0 \mathcal{L}^\dagger_1 \swss{-1}{a \omega}(\theta) e^{-i\omega u + im \varphi} \\
=& \sum_{l, m} \int_{-\infty}^\infty {\rm d} \omega \; \hat{\alpha}_{i l m}(\omega) \swsc{+1}{a \omega} \swss{+1}{a \omega}(\theta) e^{-i\omega u + im \varphi} = 0.
\end{split}
\end{equation}
Here, $\mathcal{L}^\dagger_0, \mathcal{L}^\dagger_1$ are the operators given in \cite{Chandrasekhar:1984siy}, $\lambda = \swsl{+1}{a \omega}$ is the separation constant in the angular equation for $\swss{\pm1}{a \omega}(\theta)$, and $\swsc{+1}{\nu}$ is (see \cite{Chandrasekhar:1984siy}).
\begin{equation}
\label{starobinskyconstantdef}
\swsc{+1}{\nu} = (\lambda + 1)^2 + 2 (\lambda - 1) \nu^2 + \nu^4 + 4 \nu m.
\end{equation}
The reader should note that we adopted the definition of \cite{daCosta:2019muf} for spin-weighted harmonics and constants in this article. Hence they are slightly different from those of Chandrasekhar \cite{Chandrasekhar:1984siy}.\par
The zeros of \eqref{starobinskyconstantdef} can arise only from the functional dependence of the separation constant $\lambda = \swsl{+1}{a \omega}$
on $\omega$, and $\lambda$ is fixed by the angular Teukolsky equation. This is a Sturm-Liouville type equation, and therefore the dependence on $\omega$ has to be analytic. Hence, the Starobinsky constant can vanish only at most a set of isolated points and so $\hat{\alpha}_{i l m}(\omega)=0$ except at those points, i.e. it must be a distribution that is a sum of delta-functions and its derivatives. Further, since  $\hat{\alpha}_{i l m}(\omega)$ belongs to $L^2$ w.r.t. $\omega$, this means that actually $\hat{\alpha}_{i l m}(\omega)=0$ everywhere as a distribution. Therefore, taking the inverse Fourier transform, we see that $\phi_2$ vanishes in addition to $\phi_0$.\par
Next, we intend to show that $\phi_1=0$ also. Using \eqref{MaxwelleqJl}, $\thorn \rho=\rho^2$ and that $\phi_0=\phi_2=0$ we obtain the following form for $\phi_1$:
\begin{equation}
\phi_1 = \beta^\circ \rho^2,
\end{equation}
where $\beta^\circ$ is annihilated by $\thorn$ and $\beta^\circ \circeq \left\{ -2, -2 \right\}$. By implementing the equations \eqref{MaxwelleqJm} and \eqref{MaxwelleqJmbar} one has similarly
\begin{equation}
\left( \eth - 2 \tau \right) \phi_1 = \left( \eth - 2 \tau \right) ( \rho^2 \beta^\circ ) = \rho^2 \left( \eth - 2 \frac{\bar{\rho} \tau}{\rho} \right)\beta^\circ = \bar{\rho} \rho^2 \tilde{\eth} \beta^\circ = 0,
\end{equation}
and
\begin{equation}
\left( \eth' - 2 \tau' \right) \phi_1 = \left( \eth' - 2 \tau' \right) ( \rho^2 \beta^\circ ) = \rho^2 \left( \eth' - 2 \frac{\rho \bar{\tau}}{\bar{\rho}} \right)\beta^\circ = \rho^3 \tilde{\eth}' \beta^\circ = 0.
\end{equation}
Therefore, $\tilde{\eth} \beta^\circ = \tilde{\eth}' \beta^\circ = 0$. On the other hand by equation by \eqref{MaxwelleqJmbar} we have
\begin{equation}
\left( \thorn' - 2 \rho' \right) \phi_1 = \left( \thorn' - 2 \rho' \right) \left( \rho^2 \beta^\circ \right) = \rho^2 \left( \bar{\tau} \tilde{\eth} + \tau \tilde{\eth}' + \frac{\rho'}{\rho} \thorn + \tilde{\thorn}' \right)\beta^\circ = 0.
\end{equation}
So we get $\tilde{\thorn}' \beta^\circ = \tilde{\eth} \beta^\circ = \tilde{\eth}' \beta^\circ = 0$. 
In the Kinnersley frame and advanced Kerr-Newman coordinates, this means by \eqref{eq:Hop} that $\beta^\circ$ is constant. 
Next, the fact that $A_a$ belongs to $\mathscr{V}$ implies that that the charge $\dot Q=0$. Therefore, if ${\mathscr B}$ is a cross section of ${\mathscr I}^+$, we have
\begin{equation}
0 = \dot{Q} = \frac{1}{2\pi} \int_{{\mathscr B}} \star F = \frac{1}{\pi} \Im \left\{ \int_{{\mathscr B}} \phi_1 m \wedge \widebar{m} \right\} = -4 \Re \left\{\beta^\circ \right\}.
\end{equation}
By the same argument, since $\int_{{\mathscr B}} F = 0$ for $F={\rm d}A$, we get $\Im \left\{ \beta^\circ \right\} = 0$. Thus, $\phi_1 = 0$, and therefore $A_a$ has $F_{ab}=0$.\par
Thus, we can integrate $A_a = \nabla_a \chi$, and since $A_a$ is by assumption zero near $i^0$, we can also achieve that $\chi$ is constant (e.g. zero) near $i^0$. To show that $A_a$ belongs to $\mathscr{W}$, we finally need to convince ourselves that $\chi$ is constant on the bifurcation surface $\mathscr{S}$ also.
Now, \eqref{eq:1}, \eqref{eq:2}, \eqref{eq:3}, \eqref{starH}, \eqref{Sdaggerdef}, \eqref{epsilonidentity} and $\mathcal{T}A=-\phi_0=0=\phi_1$ imply that
\begin{equation}
\Omega[A, \Re(\mathcal{S}^\dagger \Phi)] = \half \Im \int_\mathscr{S} A_{m}(\thorn + \rho)\Phi \, m \wedge \widebar{m} = 0.
\end{equation}
However, we are free to choose $\Phi, \thorn \Phi$ independently at any point of $\Sigma$, hence $\mathscr{S}$, because we a free to prescribe initial data for $\Phi$. This gives $A_m = A_{\widebar{m}} = 0$ on $\mathscr{S}$, hence $\nabla_m \chi = \nabla_{\widebar{m}} \chi = 0$ on $\mathscr{S}$ since $A_a$ is already known to be pure gauge. Hence $\chi$ is constant on $\mathscr{S}$, completing the proof that $A_a \in \mathscr{W}$.
\end{proof}

\section{Inhomogeneous Maxwell equations}
\label{sec:4}

We now investigate the retarded solution of the Maxwell equations for $A_a$ in the presence of charge and current in the Kerr-Newman background, i.e., $\mathcal{E} A_a = 4\pi J_a$, where $\mathcal{E}$ is the Maxwell operator, given in tetrad formalism in Appendix \ref{maxwellappendix}. At first, we assume that $J_a$ is smooth and of compact support in $\mathscr{M}$, and of course divergenceless. Actually, smoothness is 
just a convenient assumption to simplify the exposition, but is not really needed because the equation is linear and thus makes sense straightforwardly in the distributional sense.\par

Ideally we would like to be able to write the solution $A_a$ up to gauge in CCK form, $2\Re( \mathcal{S}^\dagger \Phi )_a$, in which $\Phi \circeq \left\{ -2, 0 \right\}$ is a suitable potential, and $\mathcal{S}^\dagger$ as defined in Appendix \ref{teukolskyoperators} -- as was possible in the case of a homogeneous equation. However, this is clearly impossible here in general since the current density corresponding to such solution automatically has a vanishing $l$-component. To solve this problem we follow an idea introduced in \cite{Green:2019nam}. The aim is to introduce a relatively simple ``corrector potential'', taken here to be of the form $xl_a$, with $x \circeq \{-1,-1\}$ to be determined, to compensate the missing $J_l$ part of current density. In the other words, we seek $x$ such that
\begin{equation}
\label{lcompensatingeq}
\left[ \mathcal{E} \left( A - xl \right) \right]^a l_a = 0.
\end{equation}
Using the GHP form of the operator $\mathcal E$ given in Appendix \eqref{maxwellappendix}, we see that we need to define $x$ by 
\begin{equation}
\label{xdef}
\rho^2 \thorn \left[ \frac{\bar{\rho}}{\rho^3} \thorn \left( \frac{\rho}{\bar{\rho}} x \right) \right] \equiv
\left( \thorn - 2\rho \right) \left( \thorn + \rho - \bar{\rho} \right) x = - 4\pi J_l.
\end{equation}
Since we want a retarded solution, 
we impose as initial conditions for this ordinary differential equation along the orbits of $l^a$ that $x$ and $\thorn x$ vanish on $\mathscr{H}^-$.
Next, we define $S_a$ as the difference of $J_a$ and the current density corresponding to $xl_a$ i.e.,
\begin{equation}
\label{sdef}
S^a \equiv 4\pi J^a - \mathcal{E}(x l)^a \Longrightarrow S_l = 0.
\end{equation}
The assumption that $J_a$ is compactly supported on the exterior region $\mathscr{M}$ of the black hole implies that there is an open neighbourhood $\mathscr{O}$ of $\mathscr{H}^-$ such that $J_a$ vanishes on $\mathscr{O}$. This together with the initial conditions for $x$ implies that $x$ and hence $S_a$ vanish on $\mathscr{O}$. We next show:
\begin{lemma}
\label{etaElemma}
Let $S_a$ be the vector field as defined above. Then there exists a complex scalar field $\eta \circeq \{ -2, 0 \}$ such that 
\begin{equation}
\label{sourcedeterminereq}
2\Re ( \mathcal{T}^\dagger \eta )_a = S_a.
\end{equation}
\end{lemma}
\begin{proof}
We define $\eta$ by \eqref{mainetaEeq}, so that the $\widebar{m}$ component of the equation \eqref{sourcedeterminereq} is satisfied,
\begin{equation}
\label{mainetaEeq}
\rho \thorn \left( \frac{\eta}{\rho} \right)\equiv\left( \thorn - \rho \right) \eta = - S_{\widebar{m}}.
\end{equation}
Then the corresponding relation for its $m$ component is automatically satisfied as well. The solution to this equation is unique up to
\begin{equation}
\label{etaEambiguity}
\eta \longrightarrow \eta + b^\circ \rho,
\end{equation}
for any $b^\circ \circeq \left\{ -3, -1 \right\}$ annihilated by $\thorn$. The $l$-component of of  \eqref{sourcedeterminereq} is trivially satisfied since neither side of  that equation has such a component. Then, only the $n$ component of relation \eqref{sourcedeterminereq} remains to be satisfied by the remaining freedom \eqref{etaEambiguity}.\par
To see what we need to choose for $b^\circ \circeq \left\{ -3, -1 \right\}$, 
we define $y$ as
\begin{equation}
S_a - 2 \Re ( \mathcal{T}^\dagger \eta )_a = y l_a \Longrightarrow y = S_n + 2 \Re \left\{ \left( \eth - \tau \right) \eta \right\}.
\end{equation}
Now, one can show that $(\mathcal{T}^\dagger \eta)^a$ is divergenceless by the definition of $\mathcal{T}^\dagger$ \eqref{Tdaggerdef},
\begin{equation}
\nabla_a \left( \mathcal{T}^\dagger \eta \right)^a = \left[ \left( \eth - \tau - \bar{\tau}' \right) \left( \thorn - \rho \right) - \left( \thorn - \rho - \bar{\rho} \right)\left( \eth - \tau \right) \right] \eta = 0,
\end{equation}
using GHP commutators and the type D property.
In addition, given that the divergence is a real operator and that $S^a$ is divergence free,
\begin{equation}
0= \nabla_a \left( S - 2 \Re \left\{ \mathcal{T}^\dagger \eta \right\} \right)^a = \nabla_a \left( y l^a \right) = \left( \thorn - \rho - \bar{\rho} \right)y 
\equiv \rho \bar \rho \thorn \left( \frac{y}{\rho \bar \rho} \right)
\end{equation}
and therefore
\begin{equation}
 y = a^\circ \rho \bar{\rho},
\end{equation}
where $a^\circ \circeq \left\{ -3, -3 \right\}$ is some real\footnote{Note that for GHP quantities of equal weight $\circeq \{p,p\}$, 
there is an invariant notion of being real, as the intrinsic conjugation operation is swapping the weights.} GHP quantity annihilated by $\thorn$.\par
If one could now solve the equation $2 \tilde{\eth} b^\circ + a^\circ = 0$ for $b^\circ \circeq \left\{ -3, -1 \right\}$, then $y$ could be set to zero by 
making use of our freedom to change $\eta \longrightarrow \eta + b^\circ \rho$. Indeed, under this change, $y$ changes as
\begin{equation}
y \longrightarrow y + 2 \Re \left\{ \left( \eth - \tau \right) b^\circ \rho \right\} = a^\circ \rho \bar{\rho} + 2 \Re \left\{ \tilde{\eth} b^\circ \right\} \rho \bar{\rho} = 0.
\end{equation}
As a result, we need to show that the equation $\tilde{\eth} b^\circ = - \half a^\circ$ is soluble.\par

To see what is going on, let us specialize momentarily  to the Reissner-Nordstrom background. There, one can expand $a^\circ$ and $b^\circ$ as follows:
\begin{subequations}
\begin{equation}
a^\circ = \sum_{l, m} \int_{-\infty}^\infty {\rm d} \omega \; \hat{a}_{l m}(\omega) \swss{0}{0}(\theta) e^{-i\omega u + im \varphi}.
\end{equation}
\begin{equation}
b^\circ = \sum_{l, m} \int_{-\infty}^\infty {\rm d} \omega \; \hat{b}_{l m}(\omega) \swss{-1}{0}(\theta) e^{-i\omega u + im \varphi}.
\end{equation}
\end{subequations}
Here, $\swss{s}{\nu}(\theta)$ are the spin-weighted spherical harmonics and $\swsl{s}{0} = l^2 + l - s^2$ are separation constants in a non-rotating (hence $\nu = 0$) black hole. Then,
\begin{equation}
\begin{split}
- 2 \tilde{\eth} b^\circ = & \sqrt{2} \sum_{l, m} \int_{-\infty}^\infty {\rm d} \omega \, \hat{b}_{l m}(\omega) \mathcal{L}^\dagger_1 \swss{-1}{0}(\theta) e^{-i\omega u + im \varphi}  \\
= & - \sum_{l, m} \, \int_{-\infty}^\infty {\rm d} \omega \, \hat{b}_{l m}(\omega) [2l(l+1)]^{1/2} \swss{0}{0}(\theta) e^{-i\omega u + im \varphi} \\
=& \sum_{l, m} \int_{-\infty}^\infty {\rm d} \omega \, \hat{a}_{l m}(\omega) \swss{0}{0}(\theta) e^{-i\omega u + im \varphi},
\end{split}
\end{equation}
where $\mathcal{L}^\dagger_1$ is the ladder operator of spin-weighted harmonics as introduced in \cite{Chandrasekhar:1984siy}. 
This gives the condition $\hat{b}_{l m}(\omega) \sqrt{2l(l+1)} = - \hat{a}_{l m}(\omega)$.
Therefore, we can solve for $\hat b_{lm}$ provided that $a^\circ$, hence $y$, has no $l=0$ part. This can easily be related to the charge $\dot Q$ of the retarded solution $A_a$, which indeed vanishes.\par
A similar, but more complicated, argument could be made for the Kerr-Newman background.
However, rather proceeding further in this way, we shall use instead the following simple argument, which relies on having a retarded solution. For this purpose, we impose the vanishing $\eta$ as an initial condition on $\mathscr{H}_-$. In addition, the fact that $S_a$ vanishes on $\mathscr{H}_-$ implies that $\eta$ vanishes on $\mathscr{H}_-$ as well. Consequently, $y$ vanishes on $\mathscr{H}_-$. On the other hand, $y = a^\circ \rho \bar{\rho}$ implies that $a^\circ$ vanishes. This implies that $y$ vanishes on the whole manifold $\mathscr{M}$ covering the exterior of the black hole, thus completing the proof.
\end{proof}

Next, we define $\Phi$ by solving the adjoint Teukolsky equation with retarded boundary conditions in the presence of the source \eqref{maxwellteukolskyeq} determined by \eqref{sourcedeterminereq},
\begin{equation}
\label{maxwellteukolskyeq}
\mathcal{O}^\dagger \Phi = \eta.
\end{equation}
Then we have:
\begin{theorem}
\label{inhomogeneoustheorem}
$A_a = xl_a + 2 \Re ( \mathcal{S}^\dagger \Phi) _a$ is a retarded solution to inhomogeneous Maxwell's equations with smooth compactly supported source $J_a$, where $\Phi$ is the retarded solution to 
\eqref{maxwellteukolskyeq}, and the corrector potential $x$ specified in \eqref{xdef}.
\end{theorem}
\begin{proof}
Given that $\mathcal{T}^\dagger \mathcal{O}^\dagger = \mathcal{E} \mathcal{S}^\dagger$ in every algebraically special spacetime and using the lemma \eqref{etaElemma}, we have
\begin{equation}
\begin{split}
& 4\pi J_a - \mathcal{E}(x l)_a = S_a \\
= & 2 \Re (\mathcal{T}^\dagger \eta) _a = 2 \Re ( \mathcal{T}^\dagger \mathcal{O}^\dagger \Phi )_a = 2 \Re (\mathcal{E} \mathcal{S}^\dagger \Phi )_a = 2 \mathcal{E} \Re ( \mathcal{S}^\dagger \Phi )_a 
\end{split}
\end{equation}
which shows that $A_a = xl_a + 2 \Re ( \mathcal{S}^\dagger \Phi) _a$ is a solution to inhomogeneous Maxwell's equations. The retarded character follows from the boundary conditions that we imposed when integrating the equation \eqref{xdef} for $x$ and \eqref{mainetaEeq} for $\eta$. 
\end{proof}

We remark that the Maxwell scalars $\phi_i$ for such $A_a = xl_a + 2 \Re ( \mathcal{S}^\dagger \Phi) _a$ are
\begin{subequations}
\begin{align}
\phi_2 =& \left( \eth' - \bar{\tau} \right) x + \tfrac{1}{2} \mathcal{O}^\dagger \Phi - \tfrac{1}{2} \eth'^2 \widebar{\Phi},\\
\phi_1 =& \tfrac{1}{2} \left( \thorn + \rho - \bar{\rho} \right) x - \tfrac{1}{4} \left[ \left( \thorn + \rho - \bar{\rho} \right) \left( \eth' + \bar{\tau} \right) + \left( \eth' + \tau' - \bar{\tau} \right) \left( \thorn + \bar{\rho} \right) \right] \widebar{\Phi},\\
\label{eq:Phi}
\phi_0 =& - \tfrac{1}{2} \thorn^2 \widebar{\Phi}.
\end{align}
\end{subequations}
In particular, we see from the equation \eqref{eq:Phi} that $\Phi$ can be obtained directly from $\phi_0$, which obeys a Teukolsky equation with source,
\begin{equation}
\label{Teuk}
\mathcal{O} \phi_0 = -4 \pi {\mathcal S} J,
\end{equation}
directly obtainable from $J_a$. However, the corrector $x$ must still be determined explicitly in order to get $A_a$.

\medskip
\noindent
\textbf{Summary:} For the benefit of the reader, we now summarize our algorithm for the forward solution of Maxwell's equations $\nabla_a F^{ab} = 4\pi J^b$ on the exterior of the Kerr-Newman black hole spacetime, with $J_a$ assumed to be of compact support:

\begin{enumerate}
\item Solve the ODE \eqref{xdef} for $x$ with vanishing initial conditions at $\mathscr{H}_-$. In the Kinnersley tetrad and advanced Kerr-Newman coordinates $(u,r,\theta,\phi_*)$, this amounts to solving
\begin{equation}
\rho^2 \partial_r \left[ \frac{\bar{\rho}}{\rho^3} \partial_r \left( \frac{\rho}{\bar{\rho}} x \right) \right] = - 4\pi J_l.
\end{equation}
with trivial initial conditions at $r=r_+$ from $r_+$ to $\infty$.

\item We define $S^a$ by eq. \eqref{sdef}, and then we find $\eta$ by integrating \eqref{mainetaEeq}. In the Kinnersley tetrad and advanced Kerr-Newman coordinates $(u,r,\theta,\phi_*)$, this amounts to solving
\begin{equation}
\rho \partial_r \left( \frac{\eta}{\rho} \right) = -S_{\widebar{m}} 
\end{equation}
with trivial initial conditions at $r=r_+$ from $r_+$ to $\infty$.

\item We solve the adjoint Teukolsky equation $\mathcal{O}^\dagger \Phi = \eta$ with retarded initial conditions. 

\item The desired ``3CK (corrected CCK)'' vector potential is obtained as $A_a =xl_a + 2 \Re ( \mathcal{S}^\dagger \Phi)_a$, or
\begin{equation}
\label{eq:CCKcorr}
A_a=xl_a+\Re \left\{ \left[ m_a \left( \thorn + \rho \right) - l_a \left( \eth + \tau \right) \right] \Phi \right\}
\end{equation}
with $\thorn,\eth,\rho,\tau$ given by the expressions in Appendix \ref{KerrNewmanQuants} referring to the Kinnersley tetrad and advanced Kerr-Newman coordinates.
\end{enumerate}

In particular, we see that the 3CK vector potential is unique as the above algorithm leaves no arbitrary choices.

\section{Singularity structure of $A_a$ in corrected 3CK form}
\label{sec:5}

Here, we briefly discuss the method in the context of the electromagnetic potential $A_a$ generated by a point-like source in the Kerr black hole background. The corresponding conserved charge-current density is
\begin{equation}
\label{pointsource}
J^a = q \int {\rm d} t (-g)^{-\sfrac{1}{2}} \dfrac{{\rm d} x_*^a}{{\rm d} t}(t) \; \delta^4(x - x_*(t)),
\end{equation}
with $x_*(t)$ some timelike geodesic in the background. This point source is evidently distributional, and not of compact support. Nevertheless, the theory developed in the previous section involves only solving linear partial differential equations and so applies also to such a distributional source if we artificially make it compactly supported by restricting the $t$ integration to a large interval $(-T,T)$ i.e. the particle only exists during this interval.\par
It is clear from the integration scheme summarized at the end of the previous section that the vector potential $A_a$ in 3CK form \eqref{eq:CCKcorr} will be non-smooth not only on the particle worldline, but also in certain locations away from the particle worldline. 
In fact, since the corrector piece $xl_a$ is obtained by integrating outward along the orbits of $l^a$ an ordinary differential equation with a $\delta$-function type source concentrated on the worldline, it is clear that $xl_a$ will have singularities (i.e. be distributional rather than smooth) on the ``shadow'' of the source, see the following fig. \ref{fig:1} taken from  \cite{Green:2019nam}.
Thus, one would expect also the 3CK $A_a$ as in \eqref{eq:CCKcorr} to have certain singularities also off the worldline.

\begin{figure}
  \includegraphics[width=0.6\textwidth,]{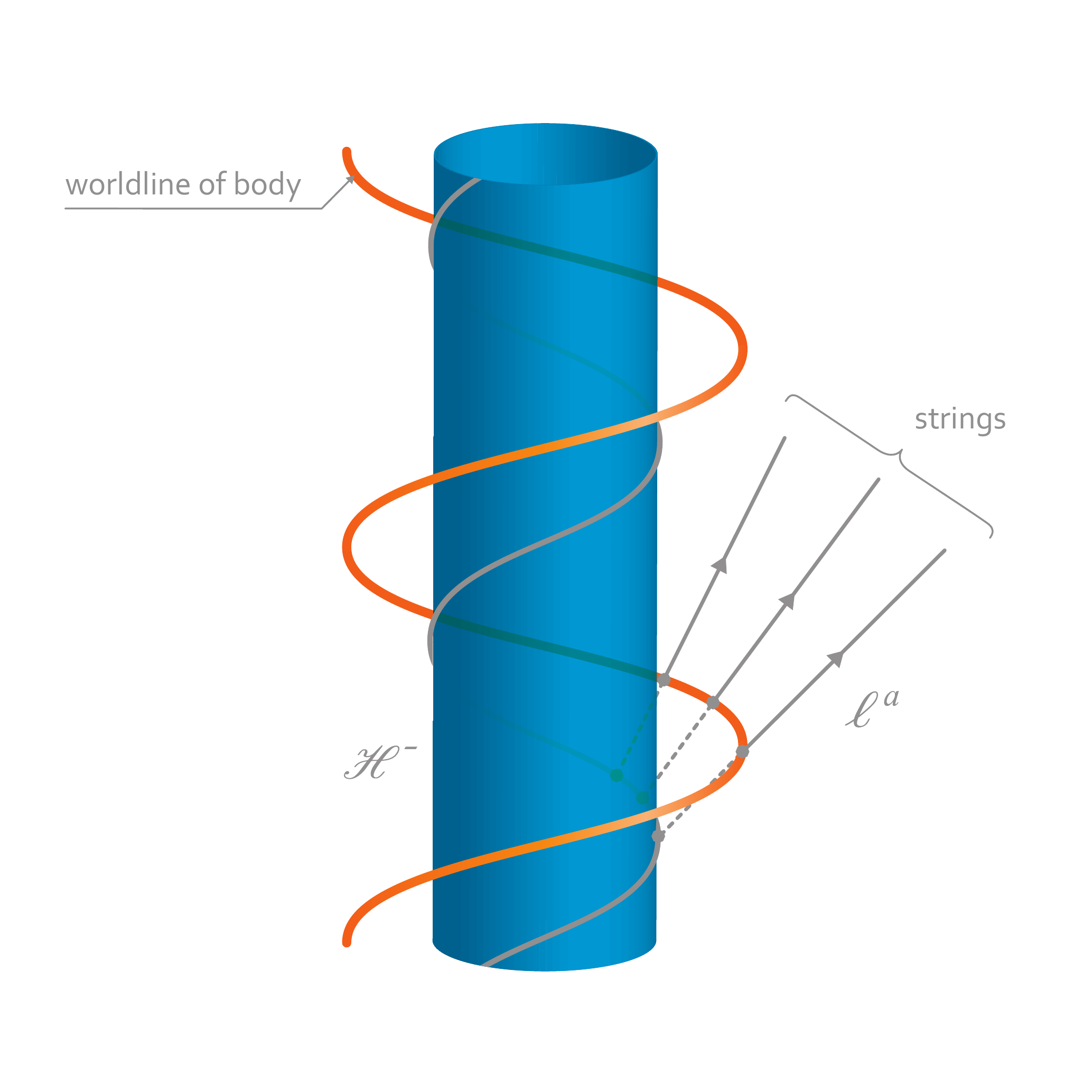}
  \caption{\label{fig:1} Worldline of a body orbiting a Kerr
    black hole. The corrector $x$ is supported on the
    semi-infinite strings (``shadows'') extending from the worldline to infinity.}
\end{figure}

We shall now analyze these singularities by comparing the 3LorenzCK vector potential \eqref{eq:CCKcorr} denoted in the following by $A_a^R$ (for ``radiation'' gauge), to a representer in Lorenz gauge denoted in the following by $A^L_a$. As $A^R_a$, the Lorenz gauge representer $A^L_a$ is required to satisfy retarded boundary conditions. It is thus uniquely determined by the wave equation
\begin{equation}
\label{eq:wave}
\nabla^b \nabla_b A_a^L  = 4\pi J_a .
\end{equation}
The gauge transformation $\chi$ connecting the two representers satisfies
\begin{equation}
A_a^R = A_a^L + \nabla_a \chi.
\end{equation}
Using \eqref{eq:CCKcorr} and GHP identities in the type D background, one shows that 
\begin{equation}
\nabla^a \nabla_a \chi = \nabla^a A_a^R = (\thorn -\rho -\bar \rho) x + 2 \Re \{ (\rho \eth - \tau \thorn) \Phi \}.
\end{equation}
Since both $A_a^R, A_a^L$ vanish by construction in a neighborhood ${\mathscr H}^- \cup {\mathscr I}^-$, this wave equation should be solved with retarded boundary conditions, and this solution is unique. The wave operator $\nabla^a \nabla_a$ acting 
on GHP scalars of trivial weight such as $\chi$ is just the spin 0 Teukolsky operator, so in practice the solution may be obtained 
by a separation of variables ansatz after a mode decomposition of the source in the equation.
\par
The singularities of $A_a^R$ off the worldline reside entirely in the gauge piece $\nabla_a\chi$ but not 
in $A_a^L$, which, as we now argue, is smooth off the worldline. Let $G$ be the retarded Green's function, with corresponding solution
\begin{equation}
A_a^L(x) = \int G_a{}^{a'}(x,x')J_{a'}(x') {\rm d}V'. 
\end{equation}
The retarded Green's function is expected to propagate the singularity in $J_a$ on the worldline only along null directions into the future. 
This expectation can be made more precise using the machinery of microlocal analysis, see e.g. \cite{hormander1964linear} for details.\par
In this calculus, the failure of the source $J_a$ to be smooth is measured by its wave front set, ${\rm WF}(J_a)$. The wave front set ${\rm WF}(u)$ of a (scalar) distribution $u$ is a subset of the cotangent space $T^*{\mathscr M}$ minus the zero section. 
Among other things, it has the property that if it is empty of over a point $x$ then $u$ is a smooth density near $x$. In the case of a 1-form valued distribution $u_a$, a refined concept is the  polarization set ${\rm WF}_{\rm pol}(u_a)$, which is a subset of $\pi^* E$, where $\pi: (x,k_a) \mapsto x$ is the projection $T^*{\mathscr M} \to \mathscr M$, see \cite{dencker1982propagation} for the precise definitions. The wave front set is in this case the union of all $(x,k_a)$ such that $(x,k_a, v_b)$ is an element of ${\rm WF}_{\rm pol}(u_b)$.\par
In our case, the definitions detailed in  \cite{hormander1964linear} easily give that  ${\rm WF}(J_b)$ is the union of $\{(x,p_a) \in T^* {\mathscr M}: \exists t \in (-T,T), x=x_*(t), p_a \dot x_*^a(t)=0, p_a \neq 0 \}$ and $\{ (x, p_a) \in T^* {\mathscr M}: x=x_*(\pm T), p_a \neq 0\}$. Since the operator $P=\nabla^a \nabla_a$ has $\sigma_P(x,p_a)=-g^{ab}(x)p_ap_b {\rm id}_{T^*{\mathscr M}}$ as its principal symbol, it automatically follows from properties of the wave front set that if $(x,p_a) \in {\rm WF}(A_b)$, then $\sigma_P(x,p_a)=0$ or $x$ is at a place where $PA_a$ is not smooth, i.e. it must be the case that either $x$ is on the worldline or $p_a$ is a null vector.\par
More detailed information can be extracted from the propagation of singularities theorem \cite{hormander1964linear}, or rather the version for partial differential operators such as $P$ operating on sections of a vector bundle $E$ over ${\mathscr M}$ \cite{dencker1982propagation}. 
The theorem states that if $(x,k_a) \notin {\rm WF}(PA_b)$, then in a neighborhood of $(x,k_a)$ in ${\mathscr N}_P = \{ (x, p_a) \in T^*{\mathscr M} \setminus 0 : \sigma_P(x,p_a) = 0 \}$, ${\rm WF}_{\rm pol}(A_b)$ is a union of Hamilton orbits for the Hamiltonian $H(x,k_a) = \sigma_P(x,k_a)$. 
More precisely, two elements of the polarization set are on the same orbit if they are equivalent, $(x,p_a,v_b) \sim (x',p_a',v_b')$, in the sense that $x$ can be joined to $x'$ by a null geodesic with tangent $p^a$ that is parallel transported to $p_a'$ and $v_a$ which is parallel transported to $v_a'$.\footnote{This follows from the fact that the Dencker connection \cite{dencker1982propagation} arising from $\nabla^a \nabla_a$ is easily checked to be the pull-back to $E$ of the Levi-Civita-connection $\nabla_a$.} \par
For a retarded solution $A_a$ which vanishes outside the causal future of the partial worldline, 
it follows from such a statement (using that the geodesic is timelike) that ${\rm WF}(A_a)$ is 
non-trivial at most either on the worldline or the null geodesic forward outflows of the beginning and end points $x_*(\pm T)$. The latter get shifted off to infinity as $T$ becomes large and play nor role. Hence, the Lorenz gauge retarded vector potential is smooth away from the worldline and the null geodesic forward outflows of the beginning and end points $x_*(\pm T)$. \par
Using the Hadamard-deWitt expansion \cite{DeWitt:1960fc} of $G$, it is possible to determine the local behavior of $A_a^L$ near the worldline. 
To this end, one defines following \cite{Merlin:2016boc} a Fermi-Walker-like coordinate system in the vicinity of the worldline. One of the coordinates, $t$, is the proper time parameter on the worldline. The other coordinates $z, x^i, i=1,2$ are Riemann normal coordinates ``perpendicular to the worldline'' defined via the exponential map set up at the point of the worldline corresponding to the chosen $t$ value. Furthermore, $z$ is defined in such a way that on the worldline, $l = (\partial_t + \partial_z)/\sqrt{2}$. The local behavior of $A_a^L$ is then given by
\begin{equation}
A_a^L = qu_a \, {\rm P.V.}\frac{1}{s} + O(\log s). 
\end{equation}
In this expression $s^2 = (x^1)^2 + (x^2)^2 + z^2$, ``P.V.'' is the principal value and $u_a(x)$ is an arbitrary smooth extrapolation of the unit tangent field of the worldline.  $O(\log)$ represents a distribution with at most logarithmic asymptotic scaling in the transversal coordinates $z, x^i, i=1,2$. Since $A_l^R=0$, the gauge function $\chi$ is then obtained integrating $\nabla_l \chi = -A_l^L$ with trivial initial conditions on ${\mathscr H}^-$. The next corollary summarizes the discussion.

\begin{corollary}
The 3CK reconstructed vector potential $A_a^R$ is gauge equivalent to a retarded vector potential in Lorenz gauge, $A_a^L$ which is smooth off the worldline (and the forward null geodesic outflows of the beginning and end-points). The unique gauge function $\chi$ relating these two is smooth except for the ``shadow'' of the worldline as in fig. \ref{fig:1}, and behaves locally near the worldline as
\begin{equation}
\chi = -q \log(s-z) + O(s \log s).
\end{equation}
\end{corollary}

Thus, we conclude that $A_a^R$ is equivalent to one of the completed ``half-string'' solutions of \cite{Pound:2013faa}. The difference in our case is that we are dealing with spin 1 and that our $A_a^R$ is manifestly a distributional solution to the inhomogeneous Maxwell equations with a systematic integration scheme for the corrector, which furthermore is shown to only have an $n$ tetrad component. 

\medskip
{\bf Acknowledgements:} S.H. is grateful to the Max-Planck Society for supporting the collaboration between MPI-MiS and Leipzig U., grant Proj. Bez. M.FE.A.MATN0003, and V.T. is grateful to International Max Planck Research School for support through a 
studentship. We have benefitted from discussions with S. R. Green, M. van de Meent, A. Pound, and P. Zimmerman. 

\noappendicestocpagenum
\begin{appendices}
\section{\label{KerrNewmanQuants}Kerr-Newman quantities}

The Kerr-Newman metric is in advanced Kerr-Newman coordinates:
\begin{multline}
{\rm d} s^{2} = \left( 1-\frac{2Mr - Q^2}{\Sigma} \right) {\rm d}u^2 + 2 {\rm d}u {\rm d}r + 2 \frac{a\sin^2\theta }{\Sigma}\left( 2Mr - Q^2 \right) {\rm d}u {\rm d}\varphi \\
- 2 a \sin^2\theta {\rm d}r {\rm d}\varphi - \Sigma {\rm d}\theta^2 + \frac{\sin^2\theta}{\Sigma} \left( \Delta a^2 \sin^2\theta - (a^2 + r^2)^2 \right) {\rm d}\varphi^2,
\end{multline}
where
\begin{subequations}
\begin{align}
\Gamma =& r + ia\cos\theta,\\
\Delta =& r^2 - 2Mr + a^2 + Q^2,\\
\Sigma =& \abs{\Gamma}^2 = r^2 + a^2 \cos^2\theta.
\end{align}
\end{subequations}
A convenient choice of Newman-Penrose tetrad $l$, $n$, and $m$ aligned with the principal null directions \cite{Geroch:1973am} is in this coordinate system
\begin{subequations}
\begin{align}
l =& \frac{\partial}{\partial r},\\
n =& \frac{r^2 + a^2}{\Sigma} \frac{\partial}{\partial u} - \frac{\Delta}{2\Sigma} \frac{\partial}{\partial r} + \frac{a}{\Sigma} \frac{\partial}{\partial \varphi},\\
m =& \frac{1}{\sqrt{2} \Gamma} \left( ia\sin\theta \frac{\partial}{\partial u} + \frac{\partial}{\partial \theta} + i\csc\theta \frac{\partial}{\partial \varphi} \right),
\end{align}
\end{subequations}
which is referred to as the ``Kinnersley frame''. In this frame,
\begin{subequations}
\begin{align}
\thorn =& \frac{\partial}{\partial r,}\\
\thorn' =& \frac{r^2 + a^2}{\Sigma} \frac{\partial}{\partial u} - \frac{\Delta}{2\Sigma} \frac{\partial}{\partial r} + \frac{a}{\Sigma} \frac{\partial}{\partial \varphi} + p \epsilon' + q \bar{\epsilon}',\\
\eth =& \frac{1}{\sqrt{2} \Gamma} \left( ia\sin\theta \frac{\partial}{\partial u} + \frac{\partial}{\partial \theta} + i\csc\theta \frac{\partial}{\partial \varphi} \right) -p \beta + q \bar{\beta}',\\
\eth' =& \frac{1}{\sqrt{2} \bar{\Gamma}} \left( -ia\sin\theta \frac{\partial}{\partial u} + \frac{\partial}{\partial \theta} - i\csc\theta \frac{\partial}{\partial \varphi} \right) + p \beta' - q \bar{\beta},
\end{align}
\end{subequations}
and
\begin{subequations}
\begin{align}
\epsilon =& 0,\\
\epsilon' =& \frac{\Delta-\bar{\Gamma} (r - M)}{2 \Sigma \bar{\Gamma}},\\
\beta =& \frac{\cot\theta}{2\sqrt{2} \Gamma},\\
\beta' =& \frac{\cot\theta}{2\sqrt{2} \bar{\Gamma}} - \frac{ia\sin\theta}{\sqrt{2} \bar{\Gamma}^2} = \bar{\beta} + \tau' =  \bar{\beta} - \frac{\rho}{\bar{\rho}} \bar{\tau},\\
\rho =& -\frac{1}{\bar{\Gamma}},\\
\rho' =& \frac{\Delta}{2 \Sigma \bar{\Gamma}},\\
\tau =& - \frac{ia\sin\theta}{\sqrt{2} \Sigma},\\
\tau' =& - \frac{ia\sin\theta}{\sqrt{2} \bar{\Gamma}^2} = -\frac{\rho}{\bar{\rho}} \bar{\tau}.
\end{align}
\end{subequations}
On any type D background satisfying $R_{ll} = R_{lm} = R_{mm} = 0$ (e.g. Kerr-Newman), we have the following identities:
\begin{subequations}
\begin{align}
\thorn \rho =& \rho^2,\\
\thorn \tau =& \rho (\tau - \bar{\tau}'),\\
\eth \rho =& \tau (\rho - \bar{\rho}),\\
\eth \tau =& \tau^2,\\
\thorn \rho' - \eth \tau' =& \bar{\rho} \rho' - \tau' \bar{\tau}' - \Psi_2.
\end{align}
\end{subequations}
Also, specifically for the Kerr-Newman background,
\begin{subequations}
\begin{align}
\thorn \tau' =& 2\rho \tau',\\
\thorn' \tau =& 2\rho' \tau,\\
\eth \rho' =& 2\rho' \tau,\\
\eth \bar{\rho} =& 2\bar{\rho} \bar{\tau}'
\end{align}
\end{subequations}
and $\kappa=\kappa'=\sigma=\sigma'=\Psi_i =0 , i \neq 2$ (or more generally for any vacuum type D metric).

\section{\label{maxwellappendix}Maxwell's equations in GHP form}

The Maxwell operator $\nabla^b F_{ab} \equiv \mathcal{E}A_b$ is on any type D background expressible as:
\begin{subequations}
\begin{multline}
\left( \mathcal{E} A \right)_l = - \left[ 2 \left( \eth' - \tau' \right) \left( \eth - \bar{\tau}' \right) - \left( \thorn - 2\rho \right) \left( \thorn' + \rho' - \bar{\rho}' \right) \right] A_l - \left( \thorn - 2\rho \right) \left( \thorn + \rho - \bar{\rho} \right) A_n\\
+ \left( \thorn - 2\bar{\rho} \right) \left( \eth' + \bar{\tau} - \tau' \right) A_m + \left( \thorn - 2\rho \right) \left( \eth + \tau - \bar{\tau}' \right) A_{\widebar{m}},
\end{multline}

\begin{multline}
\left( \mathcal{E} A \right)_n = - \left( \thorn' - 2\rho' \right) \left( \thorn' + \rho' - \bar{\rho}' \right) A_l - \left[ 2 \left( \eth - \tau \right) \left( \eth' - \bar{\tau} \right) - \left( \thorn' - 2\rho' \right) \left( \thorn + \rho - \bar{\rho} \right) \right] A_n\\
+ \left( \thorn' - 2\rho' \right) \left( \eth' + \tau' - \bar{\tau} \right) A_m + \left( \thorn' - 2\bar{\rho}' \right) \left( \eth - \tau + \bar{\tau}' \right) A_{\widebar{m}},
\end{multline}

\begin{multline}
\left( \mathcal{E} A \right)_m = - \left( \eth - 2\bar{\tau}' \right) \left( \thorn' - \rho' + \bar{\rho}' \right) A_l - \left( \eth - 2\tau \right) \left( \thorn + \rho - \bar{\rho} \right) A_n \\
+ \left[ 2\left( \thorn' - \rho' \right) \left( \thorn - \bar{\rho} \right) - \left( \eth - 2\tau \right) \left( \eth' + \tau' - \bar{\tau} \right) \right] A_m + \left( \eth - 2\tau \right) \left( \eth + \tau - \bar{\tau}' \right) A_{\widebar{m}},
\end{multline}

\begin{multline}
\left( \mathcal{E} A \right)_{\widebar{m}} = - \left( \eth' - 2\tau' \right) \left( \thorn' + \rho' - \bar{\rho}' \right) A_l - \left( \eth' - 2\bar{\tau} \right) \left( \thorn - \rho + \bar{\rho} \right) A_n \\
+ \left( \eth' - 2\tau' \right) \left( \eth' + \tau' - \bar{\tau} \right) A_m + \left[ 2\left( \thorn - \rho \right) \left( \thorn' - \bar{\rho}' \right) - \left( \eth' - 2\tau' \right) \left( \eth + \tau - \bar{\tau}' \right) \right] A_{\widebar{m}}.
\end{multline}
\end{subequations}

\section{\label{teukolskyoperators}Teukolsky operators}

The definition of adjoint of an operator defines $\pi_a, t_a, \sigma_a, w_a$ up to
a co-exact form by:
\begin{subequations}
\begin{align}
\label{pidef}
\nabla_a \pi^a(\phi, \eta) \equiv& \phi \left( \mathcal{O} \eta \right) - \eta \left( \mathcal{O}^\dagger \phi \right),\\
\label{tdef}
\nabla_a t^a(\eta, A) \equiv& \eta \left( \mathcal{T} A \right) - A_a \left( \mathcal{T}^\dagger \eta \right)^a,\\
\label{sigmadef}
\nabla_a \sigma^a(\phi, J) \equiv& \phi \left( \mathcal{S} J \right) - J_a \left( \mathcal{S}^\dagger \phi \right)^a,\\
\label{wdef}
\nabla_a w^a(A, A') \equiv& A_a \left( \mathcal{E} A' \right)^a - {A'}_a \left( \mathcal{E}^\dagger A \right)^a.
\end{align}
\end{subequations}
Here, 
\begin{subequations}
\begin{align}
\label{Odaggerdef}
\mathcal{O}^\dagger \Phi =& \left[ \left( \thorn' - \bar{\rho}' \right) \left( \thorn + \rho \right) - \left( \eth' - \bar{\tau} \right) \left( \eth + \tau \right) \right] \Phi,\\
\label{Tdaggerdef}
\left( \mathcal{T}^\dagger \eta \right)^a =& \left[ m^a \left( \thorn - \rho \right) - l^a \left( \eth - \tau \right) \right] \eta,\\
\label{Sdaggerdef}
(\mathcal{S}^\dagger \Phi)^a =& \half \left[ m^a \left( \thorn + \rho \right) - l^a \left( \eth + \tau \right) \right] \Phi,\\
\mathcal{E}^\dagger =& \mathcal{E}
\end{align}
\end{subequations}
and consequently
\begin{subequations}
\begin{align}
\label{pivalue}
\pi^a(\phi, \eta) =& l^a \phi \left( \thorn' - \rho' \right) \eta - n^a \eta \left( \thorn + \rho \right) \phi - m^a \phi \left( \eth' - \tau' \right) \eta + {\widebar{m}}^a \eta \left( \eth + \tau \right) \phi,\\
\label{tvalue}
t^a(\eta, A) =& \eta \left( m^a A_l - l^a A_m \right),\\
\label{sigmavalue}
\sigma^a(\phi, J) =& \half \phi \left( m^a J_l - l^a J_m \right),\\
\label{wvalue}
w^a(A, A') =& {F'}^{ab} A_b - F^{ab} {A'}_b = \left( \nabla^a A'^b - \nabla^b A'^a \right) A_b - \left( \nabla^a A^b - \nabla^b A^a \right) {A'}_b.
\end{align}
\end{subequations}

\section{\label{derivationofH}Derivation of $H$}

Implementing the equation $\mathcal{O} \mathcal{T} = \mathcal{S} \mathcal{E}$, one obtains the following relation \eqref{wsigmapitidentity} between $w^a$,  $\sigma^a$, $\pi^a$, and $t^a$.
\begin{equation}
\label{wsigmapitidentity}
\nabla_a \left( w^a(\mathcal{S}^\dagger \Phi, A) + \sigma^a(\Phi, \mathcal{E} A) - \pi^a(\Phi, \mathcal{T} A) - t^a(\mathcal{O}^\dagger \Phi, A) \right) = 0.
\end{equation}
As a consequence, in view of \cite{wald1990identically}, thm. 2, we learn that there must exist a 2-form $H$, constructed out of the 
fields $A_a, \Phi$ and their derivatives, such that  
$w + \sigma - \pi - t = {\rm d}\star H$, where $w = w_a(\mathcal{S}^\dagger \Phi, A) \star {\rm d}x^a$, $\sigma = \sigma_a(\Phi, \mathcal{E} A) \star {\rm d}x^a$, $\pi = \pi_a(\Phi, \mathcal{T} A) \star {\rm d}x^a$, and $t = t_a(\mathcal{O}^\dagger \Phi, A) \star {\rm d}x^a$.\par
In order to obtain an explicit form of $H$ needed in the main text, we first write $w$,  $\sigma$, $\pi$, and $t$ in GHP form.
\begin{subequations}
\begin{equation}
\begin{split}
w^a(\mathcal{S}^\dagger \Phi, A) =& \half \left( A_l n^a - A_n l^a \right) \hat{\Sigma} \Phi + \half \left( A_m \widebar{m}^a - A_{\widebar{m}} m^a \right) \hat{\Sigma} \Phi \\
+& \half \left( A_n m^a - A_m n^a \right) \left( \thorn - \rho \right) \left( \thorn + \rho \right) \Phi\\
-& \half \left( A_l \widebar{m}^a - A_{\widebar{m}} l^a \right) \left( \eth - \tau \right) \left( \eth + \tau \right) \Phi \\
+& \half \left( \phi_0 n^a - \bar{\phi}_2 l^a - \phi_1 m^a + \bar{\phi}_1 m^a \right) \left( \thorn + \rho \right) \Phi\\
-& \half \left( \phi_0 \widebar{m}^a + \bar{\phi}_0 m^a - \phi_1 l^a - \bar{\phi}_1 l^a \right) \left( \thorn + \rho \right) \Phi\\
-& \half \mathcal{O}^\dagger \Phi A_m l^a + \half \mathcal{O}^\dagger \Phi A_l m^a,
\end{split}
\end{equation}
\begin{multline}
\sigma^a(\Phi, \mathcal{E} A) = 2\pi \Phi \left( J_l m^a - J_m l^a \right) = \Phi \left( \eth - 2 \tau \right) \phi_1 l^a - \left( \thorn' - \rho' \right) \phi_0 l^a \\
+ \Phi \left( \eth' - \tau' \right) \phi_0 m^a - \Phi \left( \thorn - 2 \rho \right) \phi_1 m^a,
\end{multline}
\begin{multline}
\pi^a(\Phi, \mathcal{T} A) = - \Phi \left( \thorn' - \rho' \right) \phi_0 l^a + \phi_0 \left( \thorn + \rho \right) \Phi n^a + \Phi \left( \eth' - \tau' \right) \phi_0 m^a \\
- \phi_0 \left( \eth + \tau \right) \Phi {\widebar{m}}^a,
\end{multline}
and
\begin{equation}
t^a \left(\mathcal{O}^\dagger \Phi, A \right) = \left( m^a A_l - l^a A_m \right) \mathcal{O}^\dagger \Phi,
\end{equation}
\end{subequations}
where $\hat{\Sigma} = \eth \left( \thorn + \rho \right) - \left( \rho - \bar{\rho} \right) \left( \eth + \tau \right) = \thorn \left( \eth + \tau \right) - \left( \tau - \bar{\tau}' \right) \left( \thorn + \rho \right)$.\par
Writing $H = \half H_{a b} {\rm d}x^a\wedge {\rm d}x^b$, where $H_{ab} + H_{ba} = 0$, we 
seek to solve $w + \sigma - \pi - t = - \nabla^b H_{b a} \star {\rm d}x^a = {\rm d} \star H$. In other words,
\begin{equation}
\begin{split}
\label{unoredereddH}
- \nabla_b H^{ba} =& w^a(\mathcal{S}^\dagger \Phi, A) + \sigma^a(\Phi, \mathcal{E} A) - \pi^a(\Phi, \mathcal{T} A) - t^a(\mathcal{O}^\dagger \Phi, A)\\
=& \half \big\{ - \bar{\phi}_2 \left( \thorn + \rho \right) \Phi + \bar{\phi}_1 \left( \eth + \tau \right) \Phi + \phi_1 \left( \eth + \tau \right) \Phi + A_m \mathcal{O}^\dagger \Phi + 2\Phi \left( \eth - 2 \tau \right) \phi_1 \\
&- A_n \hat{\Sigma} \Phi + A_{\widebar{m}} \left( \eth - \tau \right) \left( \eth + \tau \right) \Phi \big\} l^a \\
+& \half \big\{ - \phi_0 \left( \thorn + \rho \right) \Phi + A_l \hat{\Sigma} \Phi - A_m \left( \thorn - \rho \right) \left( \thorn + \rho \right) \Phi \big\} n^a \\
+& \half \big\{ - \phi_1 \left( \thorn + \rho \right) \Phi + \bar{\phi}_1 \left( \thorn + \rho \right) \Phi - \bar{\phi}_0 \left( \eth + \tau \right) \Phi - A_l \mathcal{O}^\dagger \Phi\\
&- 2 \Phi \left( \thorn - 2 \rho \right) \phi_1 - A_{\widebar{m}} \hat{\Sigma} \Phi + A_n \left( \thorn - \rho \right) \left( \thorn + \rho \right) \Phi \big\} m^a \\
+& \half \big\{ \phi_0 \left( \eth + \tau \right) \Phi + A_m \hat{\Sigma} \Phi - A_l \left( \eth - \tau \right) \left( \eth + \tau \right) \Phi \big\} \widebar{m}^a.
\end{split}
\end{equation}
Next, we rewrite the divergence term in equation \eqref{unoredereddH} in GHP form.
\begin{equation}
\begin{split}
\nabla_b H^{ba} =& \left[ n^a \left( \thorn - \rho - \bar{\rho} \right) - l^a \left( \thorn' - \rho' - \bar{\rho}' \right) + \widebar{m}^a \left( \tau - \bar{\tau}' \right) - m^a \left( \tau' - \bar{\tau} \right) \right] H_{nl} \\
+& \left[ m^a \left( \thorn - \rho \right) - l^a \left( \eth - \tau \right) \right] H_{\widebar{m} n} + \left[ \widebar{m}^a \left( \thorn - \bar{\rho} \right) - l^a \left( \eth' - \bar{\tau} \right) \right] H_{m n} \\
+& \left[ m^a \left( \thorn' - \bar{\rho}' \right) - n^a \left( \eth - \bar{\tau}' \right) \right] H_{\widebar{m} l} + \left[ \widebar{m}^a \left( \thorn' - \rho' \right) - n^a \left( \eth' - \tau' \right) \right] H_{m l} \\
+& \left[ \widebar{m}^a \left( \eth - \tau - \bar{\tau}' \right) - m^a \left( \eth' - \tau' - \bar{\tau} \right) + n^a \left( \rho - \bar{\rho} \right) - l^a \left( \rho' - \bar{\rho}' \right) \right] H_{\widebar{m} m},
\end{split}
\end{equation}
where
\begin{equation}
\begin{split}
H^{ab} =& H_{nl} \left( l^a n^b - n^a l^b \right) + H_{\widebar{m} n} \left( l^a m^b - m^a l^b \right) + H_{m n} \left( l^a \widebar{m}^b - \widebar{m}^a l^b \right)\\
+& H_{\widebar{m} l} \left( n^a m^b - m^a n^b \right) + H_{m l} \left( n^a \widebar{m}^b - \widebar{m}^a n^b \right) + H_{\widebar{m} m} \left( m^a \widebar{m}^b - \widebar{m}^a m^b \right),
\end{split}
\end{equation}
which gives
\begin{equation}
\begin{split}
\label{oredereddH}
- \nabla_b H^{ba} =& \half \big\{ \left( \eth - \tau \right) \left[ 2 \phi_1 \Phi \right] - \left( \eth - \tau \right) \left[ A_n \left( \thorn + \rho \right) \Phi \right] + \left( \eth - \tau \right) \left[ A_{\widebar{m}}  \left( \eth + \tau \right) \Phi \right]\\
&+ \left( \thorn' - \rho' - \bar{\rho}' \right) \left[ A_m \left( \thorn + \rho \right) \Phi \right] - \left( \eth' - \tau' \right) \left[ A_m \left( \eth + \tau \right) \Phi \right] + \left( \rho' - \bar{\rho}' \right) \left[ A_l \left( \eth + \tau \right) \Phi \right] \big\} l^a \\
+& \half \big\{ - \left( \thorn - \rho - \bar{\rho} \right) \left[ A_m \left( \thorn + \rho \right) \Phi \right] + \left( \eth - \bar{\tau}' \right) \left[ A_l \left( \thorn + \rho \right) \Phi \right] - \left( \rho - \bar{\rho} \right) \left[ A_l \left( \eth + \tau \right) \Phi \right] \big\} n^a \\
+& \half \big\{ - \left( \thorn - \rho \right) \left[ 2 \phi_1 \Phi \right] + \left( \eth' - \tau' - \bar{\tau} \right) \left[ A_l \left( \eth + \tau \right) \Phi \right] - \left( \thorn' - \bar{\rho}' \right) \left[ A_l  \left( \thorn + \rho \right) \Phi \right] \\
&+ \left( \thorn - \rho \right) \left[ A_l \left( \thorn + \rho \right) \Phi \right] - \left( \thorn - \rho \right) \left[ A_{\widebar{m}} \left( \eth + \tau \right) \Phi \right] + \left( \tau' - \bar{\tau} \right) \left[ A_m \left( \thorn + \rho \right) \Phi \right] \big\} m^a \\
+& \half \big\{ \left( \thorn - \bar{\rho} \right) \left[ A_m \left( \eth + \tau \right) \Phi \right] - \left( \eth - \tau - \bar{\tau}' \right) \left[ A_l \left( \eth + \tau \right) \Phi \right] - \left( \tau - \bar{\tau}' \right) \left[ A_m \left( \thorn + \rho \right) \Phi \right] \big\} \widebar{m}^a.
\end{split}
\end{equation}
By integrating the equation \eqref{oredereddH}, we find $H =$  equation \eqref{H0} plus a total divergence after a very tedious calculation. 
\begin{equation}
\begin{split}
\label{H0}
H =& \left\{ \half A_m \left( \thorn + \rho \right) \Phi \right\} \left( l \wedge n \right) + \left\{ \phi_1 \Phi - \half A_n \left( \thorn + \rho \right) \Phi + \half A_{\widebar{m}} \left( \eth + \tau \right) \Phi \right\} \left( l \wedge m \right) \\
+& \left\{ - \half A_m \left( \eth + \tau \right) \Phi \right\} \left( l \wedge \widebar{m} \right) + \left\{ \half A_l \left( \thorn + \rho \right) \Phi \right\} \left( n \wedge m \right) + \left\{ \half A_l \left( \eth + \tau \right) \Phi \right\} \left( m \wedge \widebar{m} \right).
\end{split}
\end{equation}
One can simplify $\star H$  to
\begin{equation}
\label{starH}
\star H = i \left[ (\mathcal{S}^\dagger \Phi)_a {\rm d}x^a \right] \wedge A + i \phi_1 \Phi (l \wedge m) .
\end{equation}
where also
\begin{equation}
\label{epsilonidentity}
l^{[a} n^b m^c \widebar{m}^{d]} 
= \frac{-i}{4!} \varepsilon^{abcd}
\end{equation}
was used.

\end{appendices}

\end{document}